\documentclass{article}
\usepackage[margin=1in]{geometry}
\usepackage{graphicx}
\graphicspath{{BSSSnewVAM_MMS_Figs/}{../BSSSnewVAM_MMS_Figs/}{../../BSSSnewVAM_MMS_Figs/}}
\usepackage{amsmath}
\usepackage{amssymb}
\usepackage{mathtools}
\usepackage{parskip}
\usepackage[hidelinks]{hyperref}
\usepackage{natbib}
\usepackage{subcaption}
\usepackage{color}
\usepackage{ulem}
\usepackage{booktabs}
\usepackage{amsthm}
\newtheorem{lemma}{Lemma}
\newtheorem{proposition}{Proposition}
\newtheorem{remark}{Remark}
\numberwithin{equation}{section}
\numberwithin{figure}{section}
\emergencystretch=\maxdimen
\begin{document}
\title{On the brachistochrone problem for cycling ascents$\,^1$}
\author{
	Len Bos%
	\footnote{
		Universit\`a di Verona, Italy, \texttt{leonardpeter.bos@univr.it}
	}\,,
	Michael A. Slawinski%
	\footnote{
		Memorial University of Newfoundland, Canada, \texttt{mslawins@mac.com}
	}\,,\
	Rapha\"el A. Slawinski%
	\footnote{
		Mount Royal University, Canada, \texttt{rslawinski@mtroyal.ca}
	}\,,\
	Theodore Stanoev%
	\footnote{
		Memorial University of Newfoundland, Canada, \texttt{theodore.stanoev@gmail.com}
	}
}
\date{}
\maketitle
\footnote{An earlier version of this article has been published, under the same title, in \textit{Mathematics and Mechanics of Solids}, 2026, https://doi.org/10.1177/10812865261440184}
\begin{abstract}
VAM ({\it velocit\`a ascensionale media}) is a measurement that quantifies a cyclist's climbing ability.
We show that to minimize the time to attain a given height gain\,---\,which is tantamount to maximizing VAM\,---\,a cyclist should climb as steep a constant-grade hill as possible.
Apart from the power-to-weight ratio, the limit of steepness is imposed by such factors as the efficiency of pedalling, which is related to feasible cadence, maintaining balance, preventing lifting of the front, and skidding of the rear, wheel.
In an appendix, we discuss steepness constraints due to pedalling efficiency.
The article itself is focused on consequences of the power available to the cyclist, which can be viewed as a necessary condition to examine other aspects of climbing strategy.
We show that\,---\,for given start and end points, and for any fixed average-power constraint\,---\,the brachistochrone, which is the trajectory of minimum ascent time, is the straight line connecting these points, covered with a constant speed, which along such a line is equivalent to a constant power.
This is in contrast to the classical solution of a descent brachistochrone under gravity, which is a cycloid along which the speed is not constant.
\end{abstract}
\paragraph{Keywords:} 
VAM, average power, power-to-weight ratio, brachistochrone, phenomenological model
\section{Introduction}
VAM ({\it velocit\`a ascensionale media}) is the average ascent velocity.
It is a measurement of the rate of ascent that quantifies a cyclist's climbing ability.

This article is a direct follow up of \citet{BosEtAl2021}; therein, the authors reach the same conclusions but with an argument where the imposed assumptions overly reduce the candidate ascents, hence, this follow up.
Ascent optimizations are also examined by \citet{BosEtAl2024_DRNA,BosEtAl2025_arXivI,BosEtAl2025_arXivII,BosEtAl2025MMS}, and results of these four articles are used herein.
There is, however, an important distinction between \citet{BosEtAl2021} and \citet{BosEtAl2024_DRNA,BosEtAl2025_arXivI,BosEtAl2025_arXivII,BosEtAl2025MMS}.
In the 2024 and 2025 articles, we optimize the ascent strategy for a given uphill; in other words, the ascent profile is known.
In the 2021 article\,---\,as well as in the present article\,---\,we seek the profile to maximize the ascent rate for a given average power.
In other words, we seek the corresponding brachistochrone.
Notably, brachistochrone problems in relation to cycling are discussed also in other contexts~\citep[e.g.,][]{BCBC}.

We begin this article by examining the expression relating power and speed as a function of steepness.
We prove two lemmas that lead us to conclude that to minimize the time to attain a given height gain a cyclist should ride as steep a climb as possible.
We conclude the article with numerical insights and a discussion, where we comment on certain qualifiers of this conclusion.
The article includes also an appendix where we discuss steepness constraints due to pedalling efficiency.
\section{Formulation}
\label{sec:Method}
Following our previous work~\citep{BosEtAl2024_DRNA,BosEtAl2025_arXivI,BosEtAl2025_arXivII,BosEtAl2025MMS}, we let
\begin{equation}
	\label{eq:P}
   P = \dfrac{\overbrace{m\,\dfrac{{\rm d}V}{{\rm d}t}}^\text{ change in speed}+\overbrace{m\,g\sin\theta}^\text{change in elevation}+\overbrace{{\rm C_{rr}}\underbrace{m\,g\,\cos\theta}_\text{normal force}}^\text{rolling resistance}+\overbrace{\frac{1}{2}\,{\rm C_dA}\,\rho\,V^2}^\text{air resistance}}{\underbrace{1-\lambda}_\text{drivetrain efficiency}}\,V,
\end{equation}
where $P$ is the power, $V$ is the ground speed, $m$~is the mass of the bicycle-cyclist system, $g$~is the acceleration due to gravity, $\theta$ is the slope of the hill,  $\rm C_{rr}$~is the rolling-resistance coefficient, $\rm C_{d}A$~is the air-resistance coefficient, $\rho$~is the air density,  and $\lambda$ is the drivetrain-resistance coefficient.

\citet[Section~3]{BosEtAl2025_arXivI} prove that, given an average-power constraint, the ascent time is minimized if a cyclist maintains a constant ground speed, regardless of the slope.
In particular, if the average power,~$\overline{P}$\,, is constrained to a given value,~$P_0$, over the ascent time, $T$, namely,
\begin{equation}
\label{eq:L1}
\overline{P}=\frac{1}{T}\int\limits_0^TP\,{\rm d}t=P_0\,,
\end{equation}
then the optimal constant ground speed,~$V$, can be found by solving
\begin{equation}
\label{eq:OptV}
m\,g\,V\frac{H+{\rm C_{rr}}D}{L}+\frac{1}{2}\,{\rm C_dA}\,\rho\,V^3=(1-\lambda)\,P_0\,,
\end{equation}
where $H$ is the height gain, $D$ is the horizontal distance and $L$ is the arclength of the ascent curve \citep[Appendix~G]{BosEtAl2025_arXivI}.

The work done in reaching any speed\,---\,before starting the ascent\,---\,is not included in the average power for the ascent.
In any case, such an effect is negligible for longer ascents, which are required to quantify the cyclist's climbing ability using VAM.
\begin{remark}
In accordance with equation~(\ref{eq:OptV}), for fixed $H$ and $D$, the optimal constant value of $V$ is a function of $L$ alone.
\end{remark}
The ascent time is
\begin{equation}
\label{eq:T}
T=\frac{L}{V}.
\end{equation}
Given $H$ and $D$, the optimal speed,~$V=V(L)$, is defined implicitly by equation~(\ref{eq:OptV}), as a function of $L$; hence, equation~(\ref{eq:T}) expresses implicitly $T$ as a function of $L$.

Let us state the following proposition, whose proof follows from Lemma~\ref{lem:dTdL}.
\begin{proposition}
\label{prop:Straight}
For given start and end points, the curve of the minimum ascent time\,---\,for any fixed average-power constraint\,---\,is the straight line connecting these points.
\end{proposition}
\begin{lemma}
\label{lem:dTdL}
For $L>0$\,,
\begin{equation*}
\frac{{\rm d}T}{{\rm d}L}>0\,.	
\end{equation*}
\end{lemma}
\begin{proof} 
Since
\begin{equation*}
\frac{{\rm d}T}{{\rm d}L}=\frac{V-L\,V'(L)}{V^2}\,,
\end{equation*}
where $V'(L) := {\rm d}V/{\rm d}L$\,, we must show that
\begin{equation*}
V-L\,V'(L)>0\,,\qquad {\rm for}\,\,L>0\,.	
\end{equation*}
Let us write equation~(\ref{eq:OptV}) as
\begin{equation*}
 a \frac{V}{L}+b\,V^3=c\,,
 \end{equation*}
where
\begin{equation*}
a:=m\,g\,(H+{\rm C_{rr}}D)>0,\qquad b:=\frac{1}{2}{\rm C_dA}\,\rho>0,\qquad c:=(1-\lambda)\,P_0>0\,.
\end{equation*}
Differentiating implicitly with respect to $L$\,, we obtain
\begin{align*}
a\left(\frac{L\,V'-V}{L^2} \right) + 3\,b\,V^2V'&=0\\[5pt]
\implies a\,(L\,V'-V)+3\,b\,L^2V^2V'&=0\\[5pt]
\implies\,\,\, V'(a\,L+3\,b\,L^2V^2)-a\,V&=0\,,
\end{align*}
and, hence,
\begin{equation*}
V'(L)=\frac{a\,V}{a\,L+3\,b\,L^2V^2}>0\,.
\end{equation*}

Substituting this expression into $V-LV'(L)$\,, we obtain
\begin{align*}
V-LV'(L)&=V-\frac{a\,V}{a+3\,b\,L\,V^2}\\[5pt]
&=\frac{V(a+3\,b\,L\,V^2)-a\,V}{a+3\,b\,L\,V^2}\\[5pt]
&=\frac{3\,b\,L\,V^3}{a+3\,b\,L\,V^2}>0\,,
\end{align*}
as claimed.
\end{proof}
In other words, the time of ascent is an increasing function of the length of the ascent path.
Since ${\rm d}T/{\rm d}L$ is strictly greater than zero, $T$ is a strictly increasing function of~$L$ and, therefore, the value of~$L$ that minimizes~$T$ is unique.
Hence, it follows that the minimum time is attained for the curve of minimum length, which is the straight line connecting the start and end points of the ascent, as stated in Proposition~\ref{prop:Straight}.
\begin{remark}
The condition ${\rm d}T/{\rm d}L>0$ says, in words, that the longer the length of the ascent the greater the time required, and so the most direct ascent\,---\,a straight line\,---\,also requires the least amount of time.
This may appear, at first sight, to be obvious.
However, even for motion in a horizontal plane, a straight line might not result in a shortest time, as would be the case with changing rolling friction due to variable road conditions, which is akin to ray bending in accordance with Fermat's principle in layered media.
For motion in a vertical plane, the solution of the classical brachistochrone problem is a cycloid and not a straight line.
These examples suggest that, for an ascent, there are subtleties that need to be explored and thus that Lemma~\ref{lem:dTdL} is a nontrivial result.
\end{remark}
\begin{remark}
Following expressions proven in Lemma~\ref{lem:dTdL}, it follows that
\begin{equation*}
\frac{{\rm d}T}{{\rm d}L}=\frac{V-L\,V'(L)}{V^2}=\frac{3\,b\,L\,V}{a+3\,b\,L\,V^2}\,\frac{a\,L}{a\,L}=\frac{3\,b\,L^2}{a}\frac{a\,V}{a\,L+3\,b\,L^2V^2}=\frac{3\,b\,L^2}{a}\frac{{\rm d}V}{{\rm d}L}\,,
\end{equation*}
which is an interesting relation.
\end{remark}
We now consider the question of finding, among straight lines with a given height gain,~$H$, the one that gives the shortest ascent time.
The answer, as we show, is\,---\,in the limit\,---\,a vertical climb.
\begin{lemma}
\label{lem:Steep}
For $L=\sqrt{H^2+D^2}$, with a given height gain,~$H$, and a varying horizontal distance,~$D>0$\,,
 \begin{equation*}
 \frac{{\rm d}T}{{\rm d}D}>0\,.	
 \end{equation*}
\end{lemma}
\begin{proof}
Given $H$ and $L=\sqrt{H^2+D^2}$, equation~(\ref{eq:OptV}) implicitly states $V$ as a function of $D$.
Differentiating equation~(\ref{eq:T}) with respect to $D$ we obtain
\begin{equation}
\label{eq:dTdD}
\frac{{\rm d}T}{{\rm d}D}=\frac{D\,V-L^2V'}{L\,V^2},
\end{equation}
where $V' := {\rm d}V/{\rm d}D$ and we use the fact that ${\rm d}L/{\rm d}D=D/L$\,.
Letting
 \begin{equation*}
 a(D):=m\,g\,(H+{\rm C_{rr}}\,D)\,,\qquad b:=\frac{1}{2}{\rm C_dA}\,\rho>0\,,\qquad c:=(1-\lambda)\,P_0\,,
 \end{equation*}
we write equation~(\ref{eq:OptV}) as
\begin{equation*}
a(D)\frac{V}{L}+b\,V^3=c\,.
\end{equation*}
Differentiating implicitly with respect to $D$, we obtain
\begin{align*}
&a'(D)\frac{V}{L}+a(D)\left(\frac{V'(D)}{L}-\frac{V}{L^2}L'(D)\right)+3\,b\,V^2V'(D)\\[5pt]
&=a'(D)\frac{V}{L}+a(D)\left( \frac{V'(D)}{L}-\frac{V}{L^2}\frac{D}{L}\right)+3\,b\,V^2V'(D)\\[5pt]
&=V'(D)\left(\frac{a(D)}{L}+3\,b\,V^2\right)+\frac{a'(D)V}{L}-\frac{a(D)VD}{L^3}=0\,;
\end{align*}
hence,
\begin{equation*}
V'(D)\left(\frac{a(D)+3\,b\,V^2L}{L}\right)=\frac{V}{L^3}\left(-a'(D)\,L^2+a(D)\,D\right).	
\end{equation*}
But
\begin{align*}
-a'(D)\,L^2+a(D)\,D
&=-m\,g\,{\rm C_{rr}}\,(H^2+D^2)+m\,g\,(H+{\rm C_{rr}}\,D)\,D
=-m\,g\,{\rm C_{rr}}\,H^2+m\,g\,H\,D\\[5pt]
&=m\,g\,H\,(D-{\rm C_{rr}}\,H)\,.
\end{align*}
Consequently,
\begin{equation*}
V'(D)=\frac{V}{L^2}\frac{m\,g\,H\,(D-{\rm C_{rr}}\,H)}{a(D)+3\,b\,V^2L},
\end{equation*}
so that
\begin{align*}
DV-L^2V'(D)&=DV-V\frac{m\,g\,H(D-{\rm C_{rr}}H)}{a(D)+3\,b\,V^2L}\\[5pt]
&=\frac{V}{a(D)+3\,b\,V^2L}\left(D(a(D)+3\,b\,V^2L)-m\,g\,H(D-{\rm C_{rr}}H)\right)\\[5pt]
&=\frac{V}{a(D)+3\,b\,V^2L}\left(D(m\,g\,(H+{\rm C_{rr}}D))+3\,b\,V^2LD-m\,g\,HD+m\,g\,{\rm C_{rr}}H^2\right)\\[5pt]
&=\frac{V}{a(D)+3\,b\,V^2L}\left(m\,g\,{\rm C_{rr}}(D^2+H^2)+3\,b\,V^2LD\right)\\[5pt]
&=\frac{V}{a(D)+3\,b\,V^2L}\left(m\,g\,{\rm C_{rr}}L^2+3\,b\,V^2LD\right)>0
\end{align*}
and thus\,---\,by equation~(\ref{eq:dTdD})\,---\,${\rm d}T/{\rm d}D>0$\,.
\end{proof}
Consequently\,---\,according to physical considerations\,---\,a cyclist wishing to minimize the time to attain a given height gain should ride as steep a constant-grade climb as possible.

However, even though the physics suggests that the steeper the gradient the better, in the next section, we impose certain constraints and, in Appendix~\ref{app:L4}, we discuss a metabolic cost of torque fluctuations beyond which pedalling inefficiency counteracts the gains increasing the steepness of the ascent.
\section{Numerical insights}
\label{sec:NumIns}
\subsection{Speed, its vertical component and ascent time as functions of  slope}
In this section, we present numerical insights into the ground speed,~$V$, its vertical component and the ascent time, as a function of the slope,~$\theta$. 
To do so, we return to model~(\ref{eq:P}) with $P=300\,{\rm W}$, $m=68\,{\rm kg}$, $g=9.81\,{\rm m/s}^2$, $\rho=1.2\,{\rm kg/m}^3$, ${\rm C_dA}=0.3\,{\rm m}^2$, ${\rm C_{rr}}=0.005$, $\lambda=0.02$, and we let $\theta\in(5^\circ,20^\circ)$, which corresponds to slopes in the range $(8.7489\%,36.3970\%)$.%
\footnote{In cycling, as well as in road engineering, it is common to measure slope in percentages, $[\,\%\,]=\tan\theta\times 100\%$\,; thus, $45^{\circ}$ corresponds to $100\%$, which is tantamount to rise over run equal to unity.}
To solve for speed, we use the fact that expression~(\ref{eq:P}) becomes a cubic equation in $V$, with only one positive real solution~\citep[][Section~4.2]{BosEtAl2025_arXivI}.

The ground speed is shown in Figure~\ref{fig:V}.
As expected\,---\,given the power constraint\,---\,it decreases with the increase of steepness.
\begin{figure}[h]
	\centering
	\includegraphics[width=0.75\textwidth]{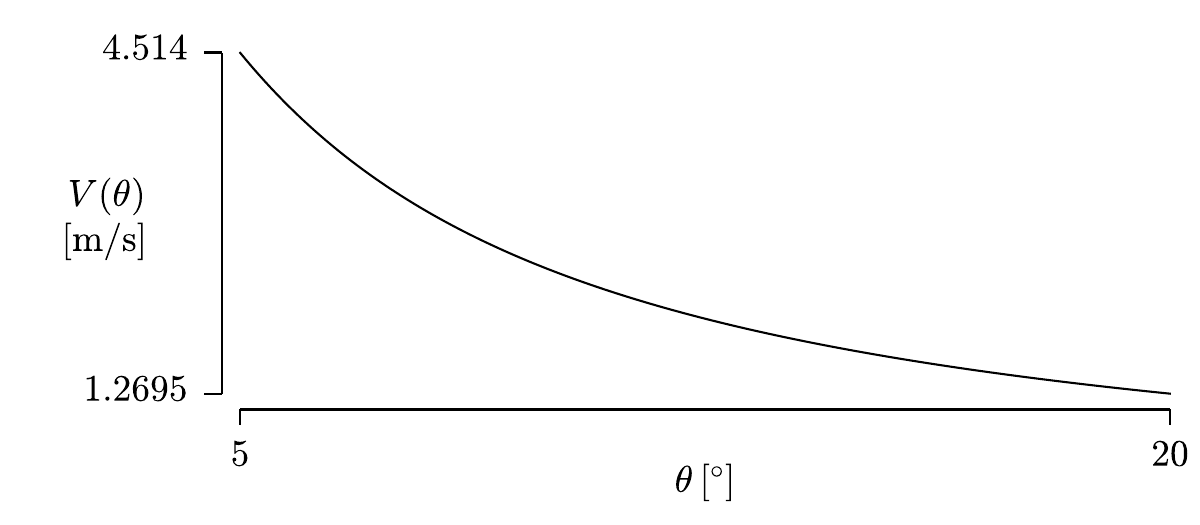}
	\caption{ $V(\theta)$: Optimal constant ground speed as a function of slope}
	\label{fig:V}
\end{figure}

The vertical component of that speed is $V(\theta)\sin\theta$, which is shown in Figure~\ref{fig:Vsin}.
In accordance with the results in Section~\ref{sec:Method}, it increases with the increase of steepness.
\begin{figure}[h]
	\centering
	\includegraphics[width=0.75\textwidth]{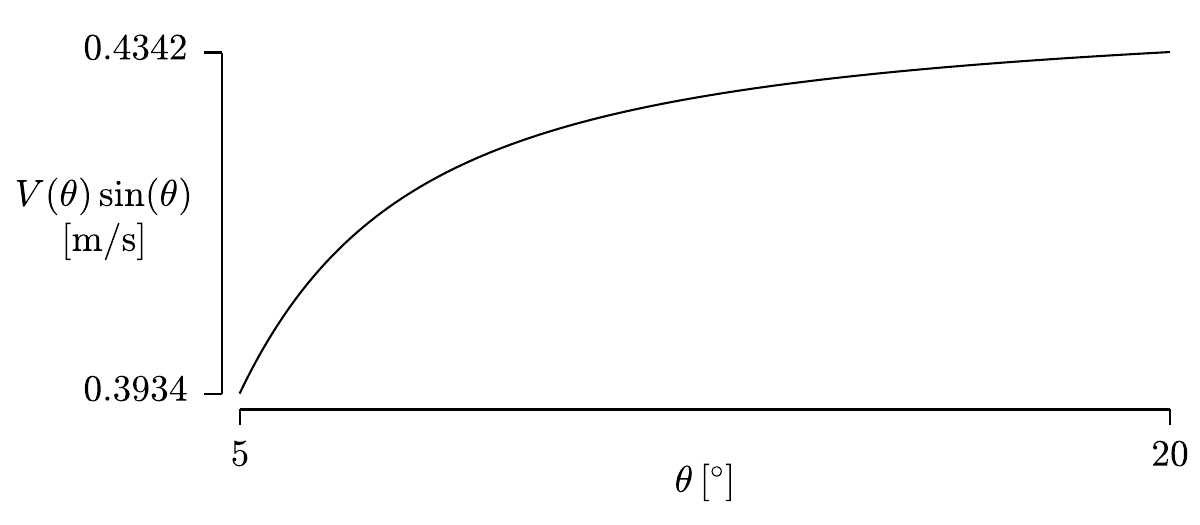}
	\caption{$V(\theta)\sin\theta$: Vertical component of speed as a function of slope}
	\label{fig:Vsin}
\end{figure}

The ascent time for the elevation gain of~$1\,000\,\rm{m}$ is shown in Figure~\ref{fig:T}.
Again in accordance with Section~\ref{sec:Method}, it decreases with steepness.
\begin{figure}
	\centering
	\includegraphics[width=0.75\textwidth]{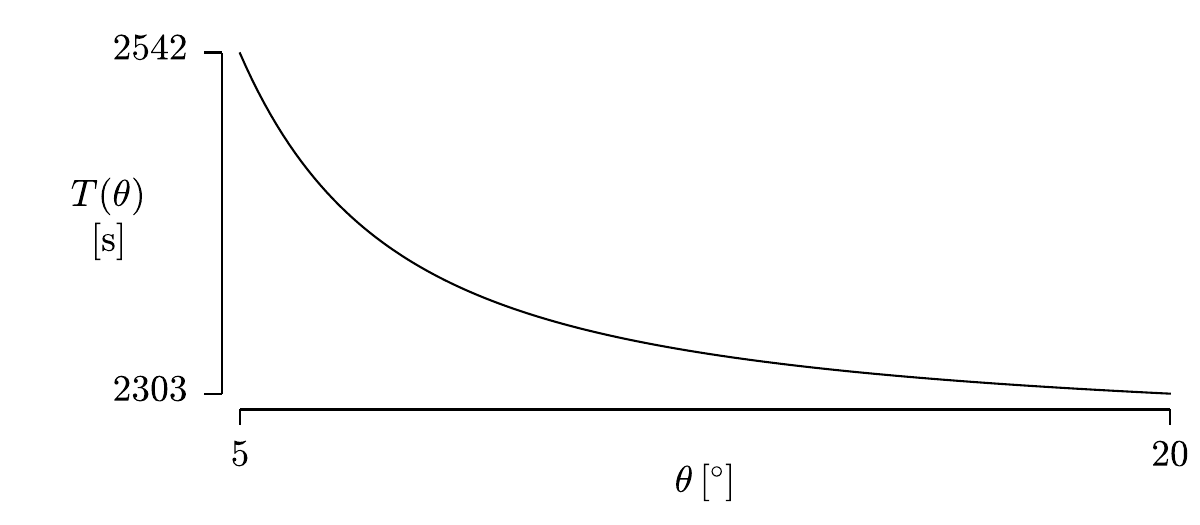}
	\caption{$T(\theta)$: Ascent time for the elevation gain of~$1\,000\,\rm{m}$ as a function of slope}
	\label{fig:T}
\end{figure}
\subsection{Relations between power, mass and optimal slopes}
\label{sub:OptimalSlopes}
In this section, we present numerical insights into relations between power, mass and optimal slopes.
Figures~\ref{fig:V}--\ref{fig:T} examine ascents under the fixed power constraint,~$P=300\,{\rm W}$, illustrating, respectively, ground speed, its vertical component, and ascent time as a function of slope.
Figure~\ref{fig:contour} examines ascents under the constraint of constant ground speed,~$V=1.5\,{\rm m/s}$.
For a given mass,~$m$, the intersection of a vertical line with a contour line representing a slope corresponds to the power,~$P$, required to achieve the ascent along that slope with~$V=1.5\,{\rm m/s}$.
As expected, an increase in mass or steepness requires greater power to maintain the constant-speed ascent.
\begin{figure}
	\centering
	\includegraphics[width=0.75\textwidth]{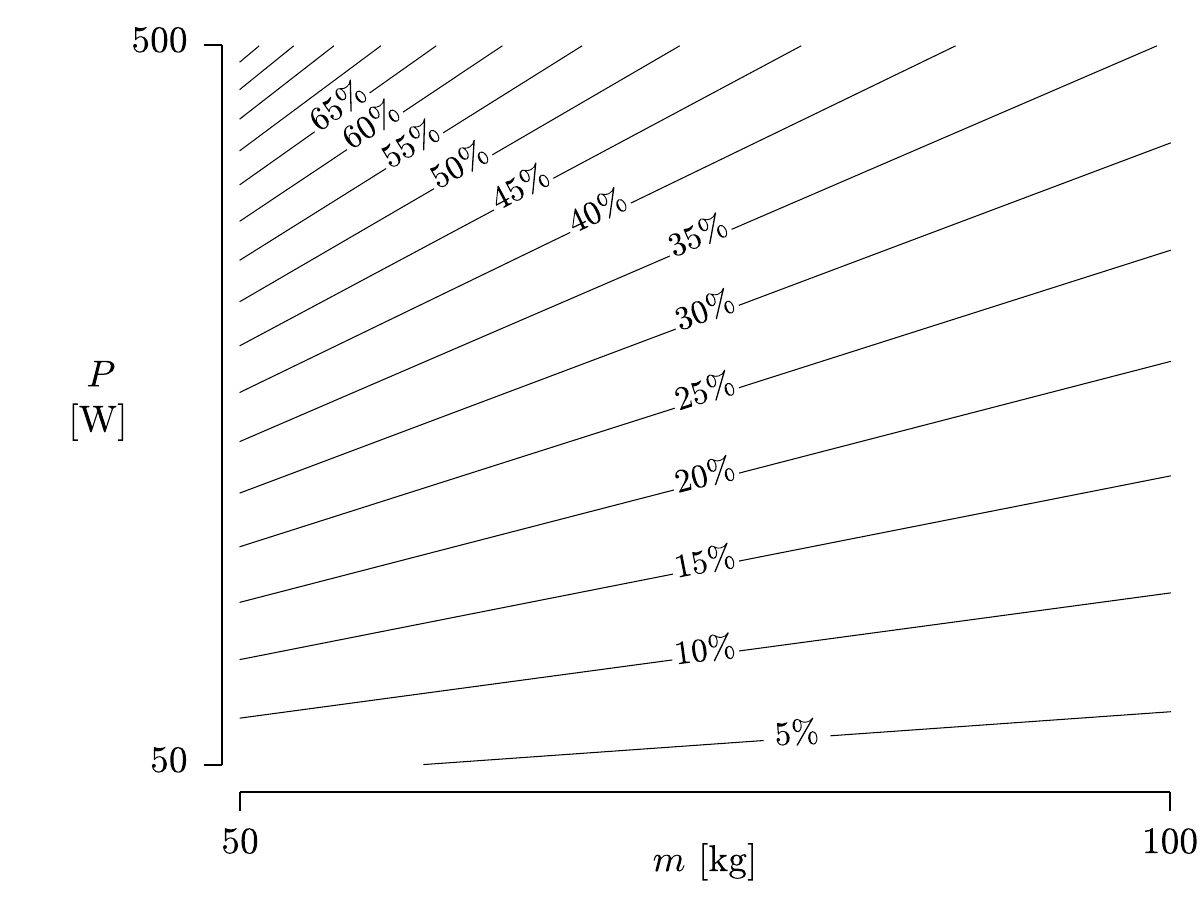}
	\caption{
	Contour lines showing maximum slope climbable with~$V=1.5\,{\rm m/s}$, for a given mass,~$m$, and power,~$P$}
\label{fig:contour}
\end{figure}

We assume $V=1.5\,{\rm m/s}$ to be the lower bound for efficiency of pedalling in uphill cycling.
For available gear ratios, it corresponds to a cadence of about 50~RPM.
Lower speeds would result in balance issues and lower cadences in an exceeding high pedal torque for each revolution.
VAM as a function of cadence, gear ratio and force is discussed by~\citet[Section~3.3]{BosEtAl2021}.

We also need to consider pragmatic limitations in handling a bicycle, such as lifting of the front wheel and skidding of the rear wheel due to steepness of a climb.
From empirical studies we infer that the steepest ascents achievable in cycling are about $30\%$; hence, the values in the upper left-hand corner of Figure~\ref{fig:contour} are unrealistic.
To gain an insight into consequences of this limitation, let us consider Figure~\ref{fig:Pmvsgrade}, where we consider a common metric in competitive cycling, namely, the power-to-weight ratio.
Riders whose power-to-weight ratio\,---\,for a given time interval\,---\,is greater than about $5\,{\rm W/kg}$ should not strive to maximize their VAM by choosing a steeper uphill but by remaining at about $30\%$ and increasing the ground speed,~$V$.
\begin{figure}
	\centering
	\includegraphics[width=0.75\textwidth]{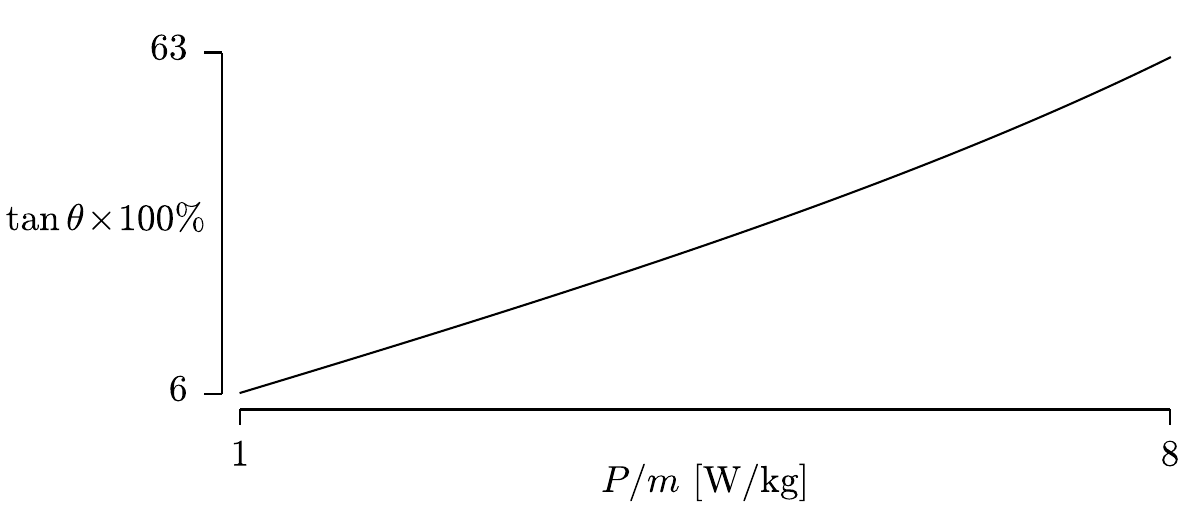}
	\caption{%
	Slope as a function of power-to-weight ratio for $V = 1.5\,\rm m/s$}
	\label{fig:Pmvsgrade}
\end{figure}

For instance, letting $P= 600\,{\rm W}$ and $m = 75\,{\rm kg}$, which corresponds to $P/m=8\,{\rm W/kg}$\,, and considering the slope of~$30\%$, as well as $\rho = 1.2\,{\rm kg/m^3}$, ${\rm CdA} = 0.3\,{\rm m^2}$, ${\rm Crr} = 0.005$, $\lambda = 0.02$, we find\,---\,according to model~(\ref{eq:P})\,---\,the ground speed $V = 2.7188\,{\rm m/s} = 9.7877\,{\rm km/h}$.
Hence, for such a rider, the optimal strategy to maximize VAM is the slope of~$30\%$ and $V = 2.7188\,{\rm m/s}$, not the slope of~$62\%$ and $V=1.5\,{\rm m/s}$, even though, as expected in view of Lemmas~\ref{lem:dTdL} and~\ref{lem:Steep}\,---\,to gain $1\,000\,{\rm m}$\,---\,$T(30\%)=1280\,{\rm s}>T(62\%)=1262\,{\rm s}$\,; the latter, however, is unrealistic.
\section{Discussion and conclusion}
\label{sec:DisCon}
Expression~(\ref{eq:L1}), which is the constraint of the optimization presented in this article, is the arithmetic mean,~$L^1$.
For a constant power, this value is identical to all generalized means,
\begin{equation*}
L^p = \left( \frac{1}{T}\int\limits_0^T |P(t)|^p\,{\rm d}t \right)^{\tfrac{1}{p}}\!\!\!, \qquad p \in [\,1, +\infty)\,;
\end{equation*}
indeed, constancy is the only condition under which all $L^p$ means are the same.
Among these, we have the quadratic mean,~$L^2$, commonly referred to as the root-mean-square, and the quartic mean,~$L^4$, known in cycling literature as Normalized Power \citep{HunterCoggan}: a nonlinear scaling of metabolic cost.
Normalized Power is a common cycling metric to account for effort variability in a manner that is not possible for average power.
$L^1$ is a linear measure; $L^4$, on the other hand, is  sensitive to power changes.
A ride alternating for equal amounts of time between, say, $200\,\rm W$ and $400\,\rm W$ still results in $L^1 = 300\,\rm W$ but in an $L^4 = 341.5\,\rm W$, quantifying the extra effort imposed by surges.

The equality, $L^1 = L^4$, which is satisfied along the brachistochrone, implies a perfectly steady effort, which can be viewed as the theoretical optimum for minimizing physiological strain empirically represented by~$L^4$, for a given mechanical output,~$L^1$.
Let us note, however, that\,---\,for a power that is a piecewise constant function\,---\,$L^p$ norms depend only on the time spent at a given power; in other words, $L^4$ remains invariant whether the power shifts once or many times.
Cyclist's fatigue, though, is likely to be different.
A choice of optimal steepness constrained by a physiological strain, represented by~$L^4$, is discussed in Appendix~\ref{app:L4}.

Under mechanical considerations, as stated by Lemmas~\ref{lem:dTdL} and \ref{lem:Steep}, respectively, the optimal ascent to maximize the VAM\,---\,for a given average power\,---\,is a constant-grade hill with the highest grade.
As proven by \citet[Section~3]{BosEtAl2025_arXivI}\,---\,for a given average power\,---\,the ascent time is minimized by a constant ground speed, regardless of slope's steepness and shape.
Herein, constancy of grade and speed imply a constancy of instantaneous power.
Constant speed combined with constant instantaneous power is the essence of this optimization.
A~constant speed ensures the shortest ascent time; the work done by the cyclist is spent entirely on overcoming external forces, with no work spent on fluctuating kinetic energy of the bicycle-cyclist system.
A~constant instantaneous power renders an attempt feasible; maintaining constant speed along variable steepness might require excessive power along steepest portions.

In other words, for given start and end points as well as a fixed average-power constraint, the brachistochrone is the straight line connecting these points.
VAM is maximized by the shortest distance covered with a constant speed corresponding to that average power, which is also the instantaneous power.
This is in contrast to the classical solution of a descent brachistochrone under gravity, which is a cycloid along which the speed is not constant.
Herein, the shortness of the straight line, which is tantamount to its steepness, is subject to the power-to-weight ratio and to pragmatic issues of pedalling efficiency and bicycle handling.
\bibliographystyle{apa}
\bibliography{BSSSnewVAM_arXiv.bib}
\begin{appendix}
\section{Steepness constraints due to pedalling efficiency}
\label{app:L4}
\setcounter{equation}{0}
As shown is Section~\ref{sec:Method}, model~(\ref{eq:P}), whose considerations are limited to physical phenomena, promotes the steepest ascent possible.
There are, however, other constraints, as stated in Section~\ref{sub:OptimalSlopes}.
An important physiological constraint is efficiency of pedalling, which we discuss in this appendix.

As the ascent steepens, pedalling becomes less smooth.
The average power per revolution begins to differ significantly from the maximum power that corresponds to a downstroke.
Following our discussion in Section~\ref{sec:DisCon}, we introduce a nonlinear metabolic cost,~$P_{\rm met}$, which we model by the~$L^4$ norm of the instantaneous power,~$P(\phi)$, where $\phi\in[\,0,2\pi\,)$.
The instantaneous power generated by the cyclist is the average power,~$P_0$, and its fluctuations of amplitude,~$A$, which depends on the steepness angle,~$\theta$\,.
We choose to model it by
\begin{equation}
\label{eq:Pphi}
P(\phi) = P_0\left(1+ A(\theta)\,\sin(2\phi)\right)\,,
\end{equation}
where $|A(\theta)|<1$\,.
During a single revolution, there are two maxima, which correspond to the downstroke of either pedal.

In a manner analogous to expression~(\ref{eq:L1}), we set
\begin{equation*}
P_0 = \,\frac{1}{2\pi}\int\limits_{0}^{2\pi}\!P(\phi)\,{\rm d}\phi\,.
\end{equation*}
To ensure that $P_0$ in expression~(\ref{eq:Pphi}) is the average power over one full revolution, as opposed to being an arbitrary scaling factor, we use this expression as the integrand to get
\begin{equation*}
\frac{1}{2\pi}\int\limits_{0}^{2\pi}\!P_0\left(1+ A(\theta)\,\sin(2\phi)\right)\,{\rm d}\phi\,
=\frac{P_0}{2\pi}\,\left[\phi-\frac{A}{2}\,\cos(2\phi)\right]^{2\pi}_0
=P_0\,,
\end{equation*}
as required.

Invoking the $L^4$ average as a metabolic cost,
\begin{equation*}
P_{\rm met} = \left[\,\frac{1}{2\pi}\int\limits_{0}^{2\pi}\!P^4(\phi)\,{\rm d}\phi\right]^{1/4}\,,
\end{equation*}
and raising the integrand\,---\,given by expression~(\ref{eq:Pphi})\,---\,to the fourth power, we get
\begin{equation*}
P_{\rm met} = P_0\,\left[\,\frac{1}{2\pi}\int\limits_{0}^{2\pi}\left(1 + 4\,A \sin(2\phi) + 6\,A^2 \sin^2(2\phi) + 4\,A^3 \sin^3(2\phi) + A^4 \sin^4(2\phi)\right)\,{\rm d}\phi\,\right]^{1/4}.
\end{equation*}
Upon integration, the odd-power sine terms vanish.
Using $\int_{0}^{2\pi}\sin^2(2\phi)\,{\rm d}\phi = \pi$ and $\int_{0}^{2\pi} \sin^4(2\phi)\,{\rm d}\phi=3\pi/4$, we simplify the expression to write
\begin{equation}
\label{eq:MetCost}
P_{\rm met} = P_0\left(\,1 + 3\,A^2(\theta) + \frac{3}{8} A^4(\theta)\,\right)^{1/4}.
\end{equation}
The ratio,
\begin{equation*}
\frac{P_{\rm met}}{P_0 } = \left(\,1 + 3\,A^2(\theta) + \frac{3}{8} A^4(\theta)\,\right)^{1/4}
\end{equation*}
is given by this formula in terms of~$A(\theta)$.
Conversely, specifying the maximum allowable ratio, we can find the corresponding maximum~$A(\theta)$,
\begin{equation}
\label{eq:P/P04}
\left(\frac{P_{\rm met}}{P_0 }\right)^4 = 1 + 3\,A^2_{\rm max} + \frac{3}{8}\,A^4_{\rm max}\,.
\end{equation}
To obtain the corresponding steepness, we let $A(\theta)$ be proportional to the vertical component of gravitational force, $A(\theta)=\kappa_0\,m\,g\,\sin(\theta)$\,, to get
\begin{equation}
\label{eq:SinMax}
\sin(\theta_{\rm max})=\frac{A_{\rm max}}{\kappa_0\,m\,g}\,,
\end{equation}
where $A_{\rm max}$ is the solution of equation~(\ref{eq:P/P04}).
The result depends only on the ratio, $P_{\rm met}/P_0$, not on the values themselves.

For a numerical example, let us use the same values as in Section~\ref{sec:NumIns}, namely, $P_0 = 300\,\rm W$, $m = 68\,\rm kg$ and $g = 9.81\,\rm m/s^2$\,.
We also let $P_{\rm met}=330\,\rm W$ and $\kappa_0=1/500\,\rm N^{-1}$.
Following expression~(\ref{eq:P/P04}), we write
\begin{equation*}
(1.1)^4 = 1 + 3\,A^2_{\rm max} + \frac{3}{8}\,A^4_{\rm max}\,.
\end{equation*}
Solving this quartic equation, we get $A_{\rm max}=0.3896$\,, and, hence, following equation~(\ref{eq:SinMax}), we obtain $\theta_{\rm max}=16.9789^\circ$\,, which is tantamount to a grade of~$30.5328\%$\,.
The same grade would be obtained with any $P_{\rm met}/P_0$ ratio that equals to~$1.1$\,.
In other words, according to model~(\ref{eq:Pphi}), $\theta_{\rm max}$ depends on the ratio of the metabolic cost and mechanical power generated by a cyclist, which is a measure of smoothness of pedalling.
Herein, following expression~(\ref{eq:Pphi}), the maximum power for each pedal downstroke is $P=416.6494\,\rm W$; the minimum power, which correspond to the vertical crank, is $P=183.3506\,\rm W$.
The instantaneous power, average power and the metabolic cost are shown in Figure~\ref{fig:PvsPhi}.
\begin{figure}[h]
	\centering
	\includegraphics[width=0.7\textwidth]{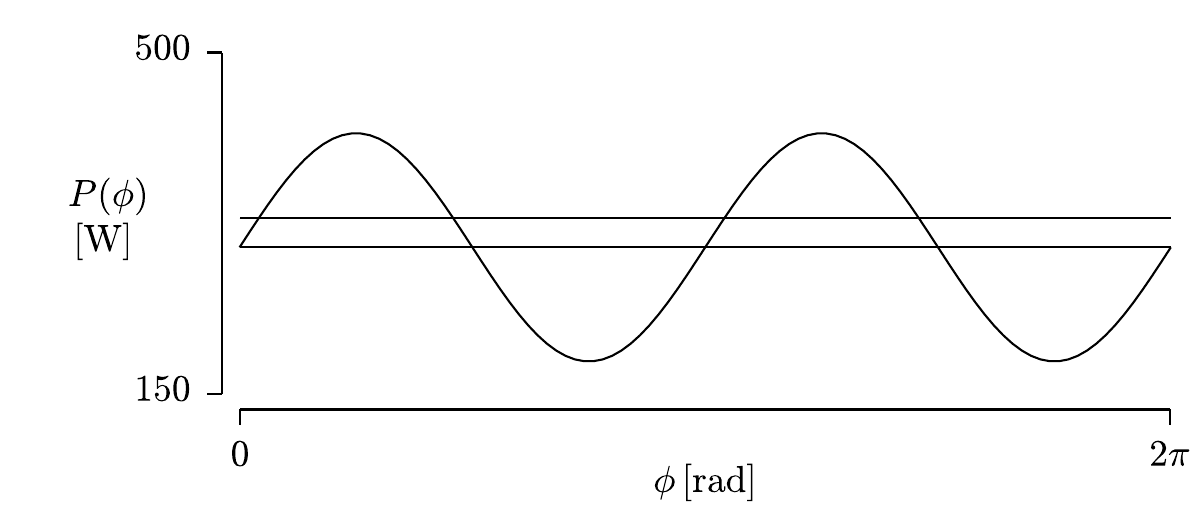}
	\caption{%
	Instantaneous power,~$P(\phi)=P_0\left(1+ 0.3896\,\sin(2\phi)\right){\rm W}$, average power,~$P_0=300{\rm W}$, and the metabolic cost,~$P_{\rm met}=330{\rm W}$, for a single pedal revolution}
	\label{fig:PvsPhi}
\end{figure}

Evaluation of these results can allow the ascent brachistochrone to be viewed as a straight line whose slope is limited by pedalling efficiency.
Empirical information to obtain $\theta_{\rm max}$ is the estimate of $P_{\rm met}/P_0$ and of $\kappa_0$\,, which depend on the cyclist and on the length of the ascent.
For further model adjustments, one might consider $\kappa_0$ as a function of~$\theta$\,.
\end{appendix}
\section*{Acknowledgements}
We wish to acknowledge insightful comments of Scott Anderson, Lucio Cadeddu, Mikhail Kochetov, Andrea Oliveri and Maurizio Vianello, as well as proofreading of David Dalton.
The research is partially supported by NSERC Discovery Grant RGPIN-2018-05158 of Michael A. Slawinski.
\end{document}